\documentclass[12pt]{article}
\usepackage[dvips]{graphicx}
\usepackage{amsmath,amsfonts,amssymb,amsthm,color}

\usepackage{natbib}

\newtheorem{theorem}{Theorem}
\newtheorem{lemma}{Lemma}
\newtheorem{corollary}{Corollary}

\newtheorem{example}{Example}

\usepackage{bm}
\bmdefine{\Bt}{t}
\bmdefine{\BX}{X}
\bmdefine{\BY}{Y}
\bmdefine{\BZ}{Z}
\bmdefine{\BB}{B}
\bmdefine{\BM}{M}
\bmdefine{\BD}{D}
\bmdefine{\Bi}{i}
\bmdefine{\Bj}{j}
\bmdefine{\Bx}{x}
\bmdefine{\By}{y}
\bmdefine{\Bz}{z}
\bmdefine{\Bv}{v}
\bmdefine{\Bw}{w}
\bmdefine{\Bn}{n}
\bmdefine{\Ba}{a}
\bmdefine{\Bb}{b}
\bmdefine{\Bc}{c}
\bmdefine{\Be}{e}
\bmdefine{\Bu}{u}
\bmdefine{\Bp}{p}
\bmdefine{\Bzero}{0}
\bmdefine{\Bone}{1}

\title{A Markov Basis for Conditional Test of Common Diagonal Effect in 
  Quasi-Independence  Model for 
  Square Contingency Tables}
\author{
Hisayuki Hara\\
Department of Technology Management for Innovation\\
University of Tokyo \\ \and
Akimichi Takemura\\
Graduate School of Information Science and Technology\\
University of Tokyo\\ \and
Ruriko Yoshida\\
Department of Statistics\\
University of Kentucky}
\date{July 2008}

\begin{document}
\maketitle

\begin{abstract}
  In two-way contingency tables we sometimes find
  that frequencies along the diagonal cells are relatively larger
  (or smaller) compared to off-diagonal cells, particularly in
  square tables with the common categories for the rows and the
  columns.  In this case the quasi-independence model with an
  additional parameter for each of the diagonal cells is usually
  fitted to the data.  A simpler model than the quasi-independence
  model is to assume a common additional parameter for all the
  diagonal cells.  We consider testing the goodness of fit of the
  common diagonal effect by Markov chain Monte Carlo (MCMC) method.  We
  derive an explicit form of a Markov basis for performing the
  conditional test of the common diagonal effect.  Once a
  Markov basis is given, MCMC procedure can be easily implemented
  by techniques of algebraic statistics.  We illustrate the
  procedure with some real data sets.
\end{abstract}

\section{Introduction}
In this paper we discuss a conditional test of a common effect for
diagonal cells in two-way contingency tables.  Modeling diagonal
effects arises mainly in analyzing contingency tables with common categories
for the rows and the columns, although our approach is applicable to
general rectangular tables.  Many models have been proposed for
square contingency tables.  \citet{tomizawa06} gives a
comprehensive review of models for square contingency tables.
\citet{gao-kuriki-2006} discuss testing marginal homogeneity
against ordered alternatives.  

Goodness of fit tests of these models are usually performed
based on the large sample approximation to the null distribution of
test statistics.  However when a model is expressed in a log-linear
form of the cell probabilities, a conditional testing procedure
(e.g.\ the Fisher's exact test for $2\times 2$ contingency tables) can
be used.  Optimality of conditional tests is a well-known classical fact
\cite[Chapter 4]{lehmann-Romano}.
Also large sample approximation may be poor when expected cell 
frequencies are small (\cite{haberman-1988-jasa}).

\citet{sturmfels1996} and \citet{diaconis-sturmfels} developed an algebraic
algorithm for sampling from conditional distributions for a
statistical model of discrete exponential families. This algorithm is
applied to conditional tests through the notion of Markov bases.  In
the Markov chain Monte Carlo approach for testing statistical fitting
of the given model, a Markov basis is a set of moves connecting all
contingency tables satisfying the given margins.
Since then  many researchers have extensively studied the
structure of Markov bases for models in computational
algebraic statistics
(e.g. \cite{hosten-sullivant, dobra-2003bernoulli, 
dobra-sullivant, geiger-meek-sturmfels,sizetwo}).

It has been well-known that for two-way contingency tables with fixed
row sums and column sums the set of square-free moves of degree two of
the form 
\[
\begin{matrix}+1 & -1 \\ -1 &
    +1\end{matrix}
\]
constitutes a Markov basis. However when we impose an additional
constraint that the sum of cell frequencies of a subtable $S$ is also
fixed, then these moves do not necessarily form a Markov basis.  
In \cite{subtable} we gave a necessary and sufficient condition on $S$
so that the set of square-free moves of degree two forms a
Markov basis.  We called this problem  a {\em subtable sum problem}.
For the common diagonal effect model defined below in (\ref{eq:cdem}) $S$
is the set of diagonal cells.  We call this problem a 
{\em diagonal sum problem}.
By the result of \cite{subtable}
we know that the  set of square-free moves of degree two does not form a 
Markov basis for the diagonal sum problem.
In this paper we give an explicit form of a 
Markov basis for the two-way diagonal sum problem.
The Markov basis contains moves of degree three and four.

When the sum of cell frequencies of a subtable $S$ is fixed to zero,
then the frequency of each cell of $S$ has to be zero and the subtable
sum problem reduces to the structural zero case.  Contingency tables
with structural zero cells are called incomplete contingency
tables (\cite[Chapter 5]{BFH1975}).
From the viewpoint of Markov bases, the subtable sum problem
is a generalization of the problem concerning structural zeros.
Properties of Markov bases for incomplete tables are
studied in \cite{aoki-takemura-2005jscs,Huber-et-al-2006,
  Rapallo-2006}.

This paper is organized as follows; In Section \ref{sec:models}, we
introduce the common diagonal effect model as a submodel of the
quasi-independence model.  In Section \ref{Preliminaries}, we
summarize some preliminary facts on algebraic statistics and Markov
bases.  Section \ref{MB} shows a Markov basis for contingency tables
with fixed row sums, column sums, and the sum of diagonal cells.
Numerical examples with some real data sets are given in Section
\ref{sec:examples}.  We conclude this paper with some remarks in
Section \ref{sec:discussion}.

\section{Quasi-Independence model and the common diagonal effect model
  for two-way contingency tables}
\label{sec:models}

Consider an $R\times C$ two-way contingency table 
$\Bx=\{x_{ij}\}$, $i=1,\dots,R$, $j=1,\dots,C$, where
frequencies along the diagonal cells are relatively larger
compared to off-diagonal cells.  Table \ref{tab:1}
\cite[Section 10.5]{agre:2002}
shows agreement between two pathologists in their diagnoses of carcinoma.
\begin{table}[ht]
\begin{center}
\caption{Diagnoses of carcinoma}
\label{tab:1}
\begin{tabular}{c|cccc}
 & 1 & 2 & 3 & 4 \\ \hline
1&  22 &  2& 2 & 0\\
2& 5& 7& 14& 0\\
3& 0& 2& 36& 0 \\
4&  0& 1& 17& 10\\
\end{tabular}
\end{center}
\end{table}
We naturally see the tendency that two pathologist agree in their
diagnoses.  Usually the quasi-independence
model is fitted to this type of data.
In the quasi-independence model, the cell probabilities $\{p_{ij}\}$
are  modeled as
\begin{equation}
\label{eq:qimodel}
\log p_{ij}= \mu+\alpha_i + \beta_j + \gamma_i \delta_{ij}, 
\end{equation}
where  $\delta_{ij}$ is Kronecker's delta.  
In (\ref{eq:qimodel}) each diagonal cell $(i,i)$, $i=1, \dots,
\min(R,C)$, has its own free parameter $\gamma_i$. This implies that
in the maximum likelihood estimation each diagonal cell is perfectly
fitted:
\[
\hat p_{ii}=  \frac{x_{ii}}{n},
\]
where $n=\sum_{i=1}^R \sum_{j=1}^C x_{ij}$ is the total frequency.

As a simpler submodel of the quasi-independence model
we consider the null hypothesis
\begin{equation}
\label{eq:cdem}
H:\ \gamma=\gamma_i, \quad i=1,\dots,\min(R,C),
\end{equation}
in the quasi-independence model.  We call this model
a {\it common diagonal effect model} and abbreviate it 
as CDEM hereafter.  In CDEM the
tendency of the diagonal cells is expressed by a single parameter,
rather than perfect fits to diagonal cells.  
We present some numerical examples of 
testing CDEM against the
quasi-independence model in Section 
\ref{sec:examples}.

Both quasi-independence models and CDEM are usually applied to square
contingency tables, i.e., $R=C$. 
As shown in Section 4, however, Markov bases of CDEM does not
essentially depend on the assumption $R=C$. 
Therefore, in this article, we consider more general cases, i.e., $R \neq C$.

Under CDEM  the sufficient statistic
consists of the row sums, column sums and the sum of the diagonal
frequencies:
\[
x_{i+}=\sum_{j=1}^C x_{ij}, \ i=1,\ldots,R, \quad
    x_{+j}=\sum_{i=1}^R x_{ij}, \  j=1,\ldots,C, \quad
 x_S = \sum_{i=1}^{\min(R,C)} x_{ii}.
\]
We write the sufficient statistic as a column vector
\[
\Bt=(x_{1+}, \dots, x_{R+}, x_{+1}, \dots, x_{+C}, x_S)'.
\]
We also order the elements of $\Bx$ lexicographically and regard $\Bx$
as a column vector.
Then with an appropriate matrix $A_S$ consisting of 0's and 1's we can write
\[
\Bt =A_S  \Bx .
\]

\section{Preliminaries on Markov bases}
\label{Preliminaries}

In this section we summarize some preliminary definitions and
notations on Markov bases (\cite{diaconis-sturmfels}).  
By now Markov bases and their uses are discussed in many papers.  See
\citet{aoki-takemura-2005jscs} for example.

The set of contingency tables ${\bm x}$ sharing the same sufficient
statistic 
\[
{\cal F}_{\bm t}= \{ {\bm x} \ge 0 \mid {\bm t}=A_S {\bm x}
\}
\]
is called a $\bm t$-fiber.  
An integer table $\Bz$ is a {\it move} for $A_S$ if
$0=A_S\Bz$.  By adding a move $\Bz$ to $\Bx\in {\cal F}_{\Bt}$, we
remain in the same fiber ${\cal F}_{\Bt}$ provided that
$\Bx + \Bz$ does not contain a negative cell.
A finite set of moves
${\cal B}=\{ \Bz_1, \dots, \Bz_L\}$ is a {\em Markov basis}, if for every $\Bt$,
${\cal F}_\Bt$ becomes connected by $\cal B$, i.e.,  we can move all
over ${\cal F}_\Bt$ by adding or subtracting the moves from ${\cal B}$
to contingency tables in ${\cal F}_\Bt$.  

If $\Bz$ is a move then $-\Bz$ is a move as well.  
For convenience we add $-\Bz$ to $\cal B$ whenever 
$\Bz \in {\cal B}$ and only consider sign-invariant Markov bases in
this paper.
A Markov
basis $\cal B$ is minimal, if every proper sign-invariant subset of
$\cal B$ is no longer a Markov basis.
A move $\Bz$ is called {\em indispensable} if $\Bz$ has to belong to
every Markov basis.  Otherwise $\Bz$ is called dispensable.  

A move $\Bz$ has positive elements and negative elements.  Separating these
elements we write 
$\Bz= \Bz^+ - \Bz^-$, where 
$(\Bz^+)_{ij}=\max(\Bz_{ij},0)$ is the positive part and
$(\Bz^-)_{ij}=\max(-\Bz_{ij},0)$ is the negative  part of $\Bz$.
$\Bz^+$ and $\Bz^-$ belong to the same fiber.

We next discuss the notion of distance reduction by a move
(\cite{aoki-takemura-2003anz, takemura-aoki-2005bernoulli,subtable}).
When $\Bx + \Bz$ does not contain a negative cell, we say that
$\Bz$ is applicable to $\Bx$.  
$\Bz$ is applicable to $\Bx$ if and only if $\Bz^- \le \Bx$
(inequality for each element).
Given two contingency tables $\Bx, \By$
let 
$
|\Bx - \By| = \sum_{i,j} | \Bx_{ij} - \By_{ij}|
$
denote the $L_1$-distance between $\Bx$ and $\By$.
For $\Bx$ and $\By$  in the same fiber, we say that $\Bz$ reduces
their distance if $\Bz$ or $-\Bz$ is applicable to $\Bx$ or $\By$ and
the distance $|\Bx - \By|$ is reduced by the application, e.g.\ 
$
|\Bx +\Bz - \By| < |\Bx -\By|.
$
A sufficient condition for $\Bz$ to reduce the distance between
$\Bx$ and $\By$ is that at least one of the following four conditions
hold:
\begin{align*}
&\text{(i) } \Bz^+ \le \Bx, \ \min(\Bz^-,\By) \neq 0, \qquad
\text{(ii) } \Bz^+ \le \By, \ \min(\Bz^-,\Bx) \neq 0, \\ 
&\text{(iii) } \Bz^- \le \Bx, \ \min(\Bz^+,\By) \neq 0, \qquad
\text{(iv) } \Bz^- \le \By, \ \min(\Bz^+,\Bx) \neq 0,
\end{align*}
where  ``$\min$'' denotes element-wise minimum.
We can also think
of reducing the distance by a sequence of moves from $\cal B$.
Clearly a finite set of moves $\cal B$ is a Markov basis if 
for every two tables  $\Bx$, $\By$ from every fiber, we can reduce the
distance $|\Bx -\By|$ by a move $\Bz$ or a sequence of moves
$\Bz_1, \dots, \Bz_k$ from $\cal B$.
We use the argument of distance reduction for proving Theorem
\ref{thm:main} in the next section.

We end this section with a known fact for the structural zero
problem.  In order to state it  we introduce two types of
moves.  In these moves, the non-zero elements are located
in the complement $S^C$
of $S$, i.e., they are in the off-diagonal cells.
\begin{itemize}
\item Type I (basic moves in $S^C$  for $\max (R, C) \geq 4$):
\[
\begin{array}{ccc}
 & j & j'\\
i & +1 & -1\\
i' & -1 & +1  \\
\end{array}  
\]
where $i,i',j,j'$ are all distinct.

\item Type II  (indispensable moves of degree 3 in $S^C$ 
for $\min (R, C) \geq 3$):
\[
\begin{array}{cccc}
  & i & i'& i'' \\
i & 0 & +1 & -1\\
i' & -1 & 0 &  +1\\
i''& +1 & -1 & 0 \\
\end{array} 
\]
where three zeros are on the diagonal.
\end{itemize}

\begin{lemma}\label{strzero}
\cite[Section 5]{aoki-takemura-2005jscs}
Moves of Type I and II form a minimal Markov basis for the structural
zero problem along the diagonal, i.e.,
$x_{ii} = 0$, $i=1,\dots,\min(R,C)$.
\end{lemma}

\section{A Markov basis for the common diagonal effect model}\label{MB}

In order to describe a Markov basis  for the diagonal sum problem, 
we introduce four additional types of moves.
\begin{itemize}
\item Type III (dispensable  moves  of degree 3 for $\min (R, C) \geq 3$):
\[
\begin{array}{cccc}
  & i & i'& i'' \\
i & +1 & 0 & -1\\
i' & 0 & -1 &  +1\\
i''& -1 & +1 & 0\\
\end{array}
\]
Note that given three distinct indices $i,i',i''$, there are three 
moves in the same fiber:
\[
\begin{array}{ccc}
 +1 & 0 & -1\\
 0 & -1 &  +1\\
-1 & +1 & 0\\
\end{array} \qquad \ 
\begin{array}{ccc}
 +1 & -1 & 0\\
-1 & 0 &  +1\\
0 & +1 & -1
\\
\end{array} \qquad \ 
\begin{array}{ccc}
 0 & -1 & +1\\
-1 & +1 &  0\\
+1 & 0 & -1\\
\end{array} 
\]
Any two of these suffice for the connectivity of the fiber.
Therefore we can choose any two moves in this fiber for minimality of
Markov basis.

\item Type IV (indispensable  moves of degree 3 for $\max (R, C) \geq 4$):
      \[
      \begin{array}{cccc}
       & i & i'& j \\
       i & +1 & 0 & -1\\
       i' & 0 & -1 &  +1\\
       j'& -1 & +1 & 0\\
      \end{array}
      \]
      where $i$, $i'$, $j$, $j'$ are all distinct. 
      We note that Type IV is similar to Type III but unlike the moves
      in Type III, the moves of Type IV are indispensable.
\item Type V  (indispensable moves of degree 4 which are non-square free):
\[
\begin{array}{cccc}
  & j & j'& j'' \\
i & +1 & +1 & -2\\
i' & -1 & -1 &  +2\\
\end{array}
\]
where $i = j$ and $i' = j'$, i.e., two cells are on the diagonal.
Note that we also include the transpose of this type as Type V moves.

\item Type VI:  (square free indispensable moves of degree 4 for $\max
      (R, C) \geq 4$): 
\[
\begin{array}{ccccc}
  & j & j'& j''& j'''' \\
i & +1 & +1 & -1 & -1\\
i' & -1 & -1 &  +1 & + 1\\
\end{array}
\]
where $i = j$ and $i' = j'$.  Type VI includes the transpose of this type.
\end{itemize}

We now present the main theorem of this paper.

\begin{theorem}
\label{thm:main}
  The above moves of Types I-VI form a Markov basis for the diagonal sum
 problem with 
  $\min(R,C)\ge 3$ and 
  $\max(R,C)\ge 4$.
\end{theorem}
\begin{proof}
Let $X, Y$ be two tables in the same fiber.  If 
\[
x_{ii} = y_{ii}, \quad \forall i=1,\dots, \min(R,C),
\]
then the problem reduces to the structural zero problem and we can
use Lemma \ref{strzero}.  
Therefore we only need to consider the difference 
\[
X-Y=Z=\{z_{ij}\}, 
\]
where there exists at least one $i$ such that $z_{ii}\neq 0$.
Note that in this case there are two indices $i \neq i'$ such that
\[
z_{ii} > 0, \qquad  z_{i'i'} <  0,
\]
because the diagonal sum of $Z$ is zero.
Without loss of generality we let $i=1$, $i'=2$.
We prove the theorem by exhausting various sign patterns of
the differences in other cells and confirming the distance reduction 
by the moves of Types I-VI.
We distinguish two cases: $z_{12} z_{21} \ge 0$ and $z_{12} z_{21} < 0$.

\medskip\noindent
{\bf Case 1} ($z_{12} z_{21} \ge 0$): 
In this case without loss of generality
assume that $z_{12}\ge 0$, $z_{21} \ge 0$. Let $0+$ denote the  cell 
with non-negative value of $Z$ 
and let $*$ denote a cell with arbitrary value of $Z$.
Then $Z$ looks like
\[
\begin{array}{cccc}
+ & 0+ & * & \cdots \\
0+ & - & * & \cdots\\
* & * & *  & \cdots\\
\vdots & \vdots & \vdots & \ddots
\end{array}
\]
Note that there has to be a negative cell on the first row and on the
first column.  
Let $z_{1j}<0$, $z_{j'1}<0$.
Then $Z$ looks like
\[
\begin{array}{ccccccc}
  & 1 & 2  & \cdots & j & \cdots & \\
1 & + & 0+ & \cdots & - & \cdots & \\
2 & 0+ & - & \cdots &  * & \cdots & \\
 & \vdots & \vdots & \vdots &  \vdots & \cdots &\\
j' & - & * & \cdots & * & \cdots & \\
 & \vdots & \vdots & \vdots &  \vdots & \ddots &\\
\end{array}
\]
If $j=j'$, we can apply a Type III move to reduce the $L_1$ distance. 
If $j \neq j'$, we can apply a Type IV move to reduce the $L_1$
distance.   
This takes care of the case $z_{12}z_{21}\ge 0$.

\medskip
\noindent
{\bf Case 2} ($z_{12} z_{21} < 0$): 
Without loss of generality
assume that $z_{12} > 0$, $z_{21} <  0$.  Then $Z$ looks like
\[
\begin{array}{cccc}
+ & + & * & \cdots \\
- & - & * & \cdots \\
* & * & * & \cdots \\
\vdots & \vdots & \vdots & \ddots \\
\end{array}
\]
There has to be a negative cell on the first row and there has to be a
positive cell on the second row.
Without loss of generality we can let
$z_{13}<0$ and at least one of $z_{23}, z_{24}$ is positive. Therefore
$Z$ looks like
\begin{equation}
\label{eq:cases2122}
\begin{array}{cccccc}
+ & + & - & * &  *& \cdots \\
- & - & * & + & * & \cdots \\
* & * & * & * & * &\cdots \\
\vdots & \vdots & \vdots & \vdots& \vdots & \ddots \\
\end{array}
\quad \text{or}\quad 
\begin{array}{ccccc}
+ & + & - & * & \cdots \\
- & - & + & * & \cdots \\
* & * & * & * & \cdots \\
\vdots & \vdots & \vdots & \vdots& \ddots \\
\end{array}
\end{equation}

These two cases are not mutually exclusive.  We look at $Z$ as the
left pattern whenever possible. Namely, whenever we can find two
different columns $j, j'\ge 3$, $j\neq j'$ such that
$z_{1j}z_{2j'}<0$, then we consider $Z$ to be of the left pattern.  We
first take care of the case that $Z$ does not look like the left
pattern of (\ref{eq:cases2122}), i.e., there are no $j, j'\ge 3$,
$j\neq j'$, such that $z_{1j}z_{2j'}<0$.

\medskip
\noindent
{\bf Case 2-1} ($Z$ does not look like the left pattern of
(\ref{eq:cases2122})): 
If there exists some $j\ge 4$  such that $z_{1j} < 0$, then
in view of $z_{23}>0$ we have $z_{1j} z_{23}<0$ and $Z$ looks like the
left pattern of (\ref{eq:cases2122}).
Therefore we can assume
\[
z_{1j} \ge 0, \quad \forall j\ge 4.
\]
Similarly
\[
z_{2j} \le 0, \quad \forall j\ge 4
\]
and  $Z$ looks like
\[
\begin{array}{cccccc}
+ & + & - & 0+ & \cdots &0+ \\
- & - & + & 0- & \cdots &0-\\
* & * & * & * & \cdots &*\\
\vdots & \vdots & \vdots & \vdots& \vdots& \vdots \\
\end{array}
\]
Because the first row and the second row sum to zero, we have
\[
z_{13} \le -2, \quad z_{23} \ge 2.
\]
However then we can apply Type V move to reduce the $L_1$ distance.

\medskip
\noindent
{\bf Case 2-2} ($Z$ looks like the left pattern of
(\ref{eq:cases2122})): 
Suppose that there exists some $i\ge 3$ such that
$z_{i3}>0$.  If  $z_{33}>0$, then $Z$ looks like
\[
\begin{array}{cccccc}
+ & + & - & * &  *& \cdots \\
- & - & * & + & * & \cdots \\
* & * & + & * & * &\cdots \\
* & * & * & * & * &\cdots \\
\vdots & \vdots & \vdots & \vdots& \vdots & \ddots \\
\end{array}
\]
Then we can apply a type III move involving
\[
z_{12}>0,\ z_{13}<0, \ z_{22}<0, \ z_{24}>0, \ z_{33}>0, \ z_{34}:\text{arbitrary}
\]
and reduce the $L_1$ distance.  On the other hand if $z_{i3}>0$ for
$i\ge 4$, then 
$Z$ looks like
\[
\begin{array}{cccccc}
+ & + & - & * &  *& \cdots \\
- & - & * & + & * & \cdots \\
* & * & * & * & * &\cdots \\
* & * & + & * & * &\cdots \\
* & * & * & * & * &\cdots \\
\vdots & \vdots & \vdots & \vdots& \vdots & \ddots \\
\end{array}
\]
Then  we can apply a type IV move involving
\[
z_{11}>0,\ z_{13}<0, \ z_{21}<0, \ z_{24}>0, \ z_{i3}>0, \ z_{ii}:\text{arbitrary}
\]
and reduce the $L_1$ distance.  Therefore we only need to consider $Z$
which looks like
\[
\begin{array}{cccccc}
+ & + & - & * &  *& \cdots \\
- & - & * & + & * & \cdots \\
* & * & 0- & * & * &\cdots \\
\vdots & \vdots & \vdots & \vdots& \vdots & \cdots \\
* & * & 0- & * & * &\cdots \\
\end{array}
\]
Similar consideration for the fourth column of $Z$ forces
\[
\begin{array}{cccccc}
+ & + & - & * &  *& \cdots \\
- & - & * & + & * & \cdots \\
* & * & 0- & 0+ & * &\cdots \\
\vdots & \vdots & \vdots & \vdots& \vdots & \cdots \\
* & * & 0- & 0+ & * &\cdots \\
\end{array}
\]
However then because the third column and the fourth column sum to zero, 
we have $z_{23}>0$ and $z_{14}<0$ and $Z$ looks like
\[
\begin{array}{cccccc}
+ & + & - & - &  *& \cdots \\
- & - & + & + & * & \cdots \\
* & * & 0- & 0+ & * &\cdots \\
\vdots & \vdots & \vdots & \vdots& \vdots & \cdots \\
* & * & 0- & 0+ & * &\cdots \\
\end{array}
\]
Then we apply Type VI move to reduce the $L_1$ distance.

Now we have exhausted all possible sign patterns of $Z$ and 
shown that the $L_1$ distance can always be decreased by some move of
Types I-VI.
\end{proof}

Since moves of Type I, II, IV, V and VI are indispensable, we have the
following corollary.

\begin{corollary}
 \label{col:1}
 A minimal Markov basis for the diagonal sum problem with 
 $\min(R,C)\ge 3$ and $\max(R,C)\ge 4$ 
 consists of moves of Types I, II, IV, V, VI
 and two moves of Type III for each given triple  $(i,i',i'')$.
\end{corollary}

\section{Numerical examples}
\label{sec:examples}

In this section with the Markov basis computed in previous sections, we will
experiment via MCMC method.  Particularly, we test the hypothesis 
of CDEM for a given data set.

Denote expected cell frequencies under the quasi-independence
model and CDEM by
\[
\hat m^{QI}_{ij}=n \hat p^{QI}_{ij}, \qquad
\hat m^{S}_{ij}=n \hat p^{S}_{ij}, 
\]
respectively.  These expected cell frequencies can be computed
via the iterative  proportional fitting (IPF).  
IPF for the quasi-independence model is explained in 
Chapter 5 of \cite{BFH1975}.  
IPF for  the common diagonal effect model is given as follows.
The superscript $k$ denotes the step count.
\begin{enumerate}
\item Set $m^{S,k}_{ij} = m_{ij}^{S,k-1} x_{i+}/m_{i+}^{S,k-1}$ 
for all $i,\, j$ and
set $k = k + 1$.  Then go to Step 2.
\item Set $m^{S,k}_{ij} = m_{ij}^{S,k-1} x_{i+}/m_{i+}^{S,k-1}$ for all $i,\, j$ and
set $k = k + 1$.  Then go to Step 3.
\item Set $m^{S,k}_{ii} = m_{ii}^{S,k-1} x_S/m^{S,k-1}_S$ for all $i=1,\dots,\min(R,C)$ and
$m^{S,k}_{ij} = m_{ij}^{S,k-1} (n-m^{S,k-1}_S)/(n-x_S)$ for all $i \not = j$. 
Then set $k = k + 1$ and go to Step 1.
\end{enumerate}
After convergence  we set 
\[
\hat m^S_{ij} = m_{ij}^{S,k} \text{ for all }i, \, j.
\]
We can initialize $m^{S,0}$ by 
\[
m_{ij}^{S,0} = n/(R\cdot C)\text{ for all } i, \, j.
\]

As the discrepancy measure from the 
hypothesis of the common diagonal model, 
we calculate ($2 \times$)  the log likelihood ratio statistic
\[
G^2 = 2 \sum_i \sum_j x_{ij} \log\frac{\hat m^{QI}_{ij}}{\hat m^S_{ij}}.
\]
for each sampled table ${\bm x}=\{x_{ij}\}$.  

In all experiments in this paper, we sampled 10,000 tables after 
8,000 burn-in steps.

\begin{example}
The first example is from Table \ref{tab:1} of Section \ref{sec:models}.
The value of $G^2$ for the observed
table in Table \ref{tab:1} is $13.5505$ and the corresponding asymptotic 
$p$-value is  $0.003585$ from the asymptotic distribution $\chi_3^2$.

A histogram of sampled tables via MCMC with a Markov basis for Table \ref{tab:1}
is in Figure \ref{histgram}.  We estimated the p-value $0.00379$ via MCMC
with the Markov basis computed in this paper.  Therefore CDEM model is
rejected at the significance level of 5\%.

\begin{figure}[ht!]
\begin{center}
     \includegraphics[scale= 0.6]{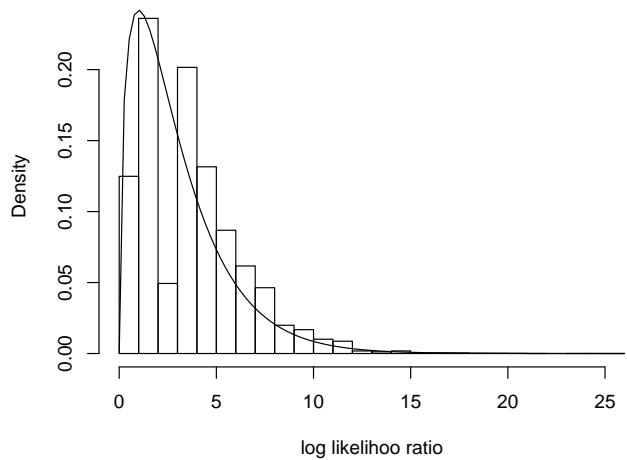}
 \end{center}
\caption{A histogram of sampled tables via MCMC with a Markov basis computed
for Table \ref{tab:1}.  The black line shows the asymptotic distribution $\chi_3^2$.}\label{histgram}
\end{figure}

\end{example}

\begin{example}
The second example is Table 2.12  from \cite{agre:2002}.  
\begin{table}[ht!]
\begin{center}
\caption{Married couples in Arizona}
\label{tab:2}
\begin{tabular}{c|cccc}
& never/occasionally& fairly often &very often &almost always\\\hline
never/occasionally & 7& 7& 2 &3\\
fairly often & 2& 8 &3 &7\\
very often & 1& 5& 4 &9\\
almost always & 2& 8& 9 &14\\
\end{tabular}
\end{center}
\end{table}
Table \ref{tab:2} summarizes responses of 91 married couples in Arizona about 
how often sex is fun. Columns represent wives' responses and rows represent
husbands' responses.

The value of $G^2$ for the observed
table in Table \ref{tab:2} is $6.18159$ and the corresponding asymptotic 
$p$-value is  $0.1031$ from the asymptotic distribution $\chi_3^2$.

A histogram of sampled tables via MCMC with a Markov basis for Table \ref{tab:2} 
is in Figure \ref{histgram2}.  We estimated the p-value $0.12403$  via MCMC
with the Markov basis computed in this paper.
Therefore CDEM model is
accepted at the significance level of 5\%.  We also see that $\chi_3^2$
approximates well with this observed data.
\begin{figure}[ht!]
\begin{center}
     \includegraphics[scale= 0.6]{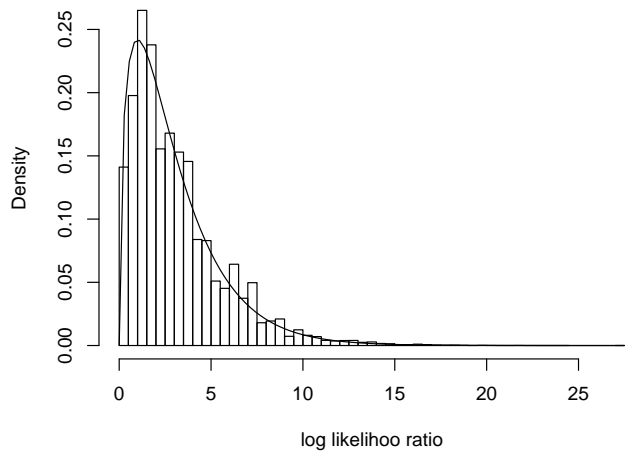}
 \end{center}
\caption{A histogram of sampled tables via MCMC with a Markov basis computed
for Table \ref{tab:2}.  The black line shows the asymptotic distribution $\chi_3^2$.}\label{histgram2}
\end{figure}
\end{example}

\begin{example}
 \label{ex:3}
The third example is Table 1 from \cite{diaconis-sturmfels}.  
\begin{table}[ht!]
\begin{center}
\caption{Relationship between birthday and death day}
\label{tab:3}
\begin{tabular}{c|cccccccccccc}
 & Jan & Feb & March & April & May & June & July & Aug & Sep & Oct & Nov & Dec \\\hline
Jan & 1& 0 &0 &0 &1 &2 &0 &0 &1 &0& 1 &0\\
Feb & 1& 0 &0 &1 &0& 0& 0 &0 &0 &1& 0 &2\\
March & 1& 0& 0& 0& 2& 1& 0 &0& 0& 0& 0& 1\\
April & 3& 0 &2 &0 &0 &0& 1 &0 &1 &3& 1& 1\\
May & 2 &1& 1& 1 &1 &1& 1& 1& 1 &1 &1& 0\\
June & 2& 0& 0 &0 &1& 0 &0 &0& 0& 0 &0& 0\\
July & 2& 0 &2 &1& 0& 0 &0 &0& 1& 1& 1& 2\\
Aug & 0& 0 &0 &3 &0& 0 &1 &0 &0 &1& 0 &2\\
Sep & 0& 0 &0 &1& 1& 0 &0 &0 &0& 0 &1 &0\\
Oct & 1& 1 &0& 2& 0 &0& 1& 0 &0& 1 &1& 0\\
Nov & 0& 1 &1& 1& 2 &0 &0& 2 &0 &1& 1& 0\\
Dec & 0& 1& 1& 0 &0 &0 &1& 0& 0& 0 &0 &0\\
\end{tabular}
\end{center}
\end{table}
Table \ref{tab:3} shows data gathered to test the hypothesis of association 
between birth day
and death day. The table records
the month of birth and death for 82 descendants of Queen Victoria. A widely
stated claim is that birthday-death day pairs are associated. Columns represent
the month of birth day and rows represent the month of death day.
As discussed in \cite{diaconis-sturmfels}, the Pearson's $\chi^2$
statistic for the usual independence model is 115.6 with 121 degrees
of freedom.  Therefore the usual independence model is accepted for
this data.  However, when CDEM is fitted,  the Pearson's $\chi^2$
becomes 111.5 with 120 degrees of freedom.  Therefore the fit
of CDEM is better than the usual independence model.

We now test CDEM against the quasi-independence model.
The value of $G^2$ for the observed
table in Table \ref{tab:3} is $6.18839$ and the corresponding asymptotic 
$p$-value is $0.860503$ from the asymptotic distribution $\chi_{11}^2$.

A histogram of sampled tables via MCMC with a Markov basis for Table \ref{tab:3}
is in Figure \ref{histgram3}.  We estimated the p-value $0.89454$ via MCMC
with the Markov basis computed in this paper.  There exists a large
discrepancy between the asymptotic distribution and the distribution
estimated by MCMC due to the sparsity of the table.
\begin{figure}[ht!]
\begin{center}
     \includegraphics[scale= 0.6]{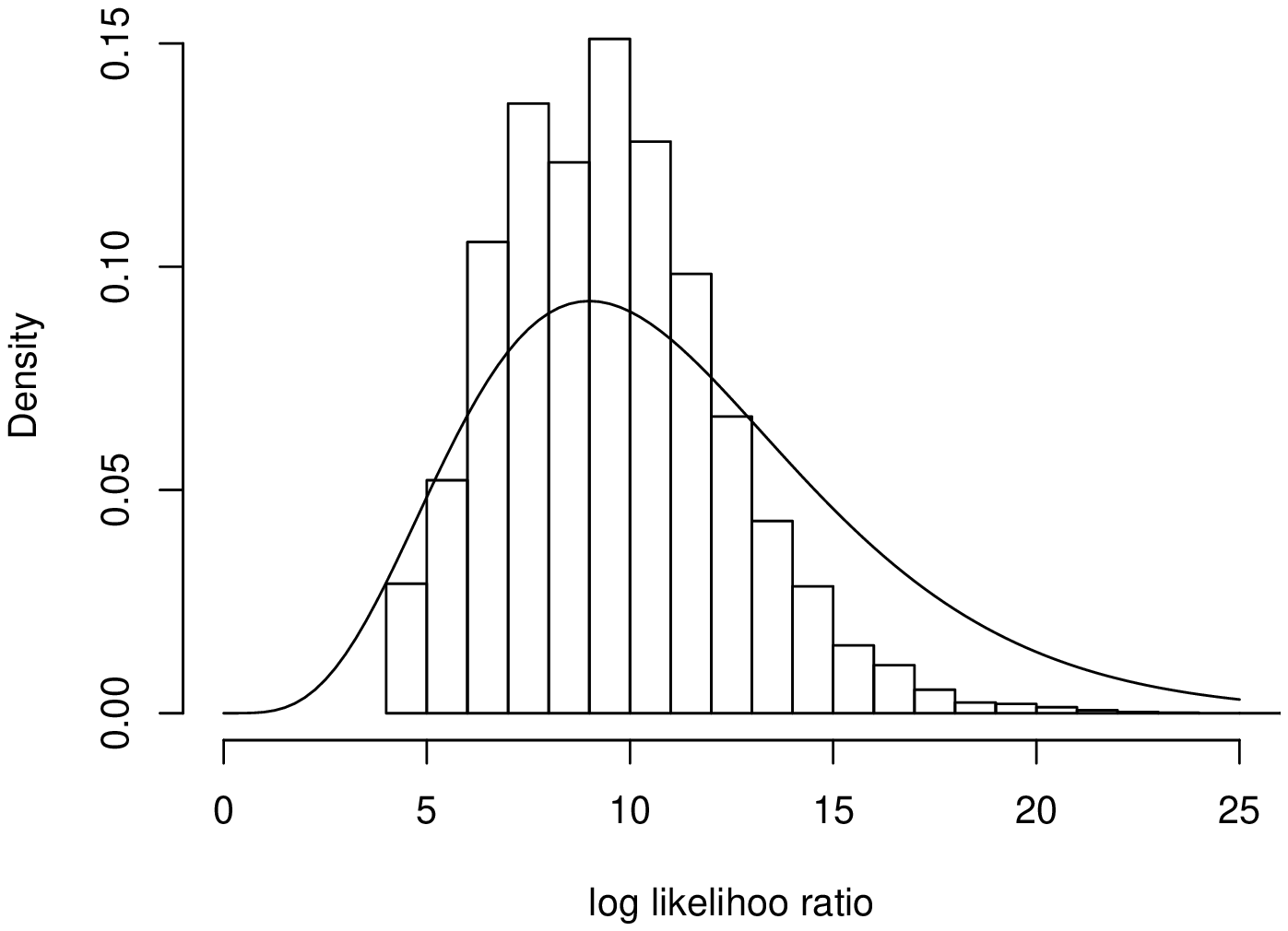}
 \end{center}
\caption{A histogram of sampled tables via MCMC with a Markov basis computed
for Table \ref{tab:3}.  The black line shows the asymptotic distribution $\chi_{11}^2$.}\label{histgram3}
\end{figure}
\end{example}

\section{Concluding remarks}
\label{sec:discussion}

In this paper we derived an explicit form of a Markov basis for the
diagonal sum problem.  With this Markov basis we showed that we can
easily run the conditional test of the common diagonal effect model.
As seen from Figure \ref{histgram3} in Example \ref{ex:3}, 
there may exist a large discrepancy between the asymptotic distribution 
and the distribution estimated via MCMC. 
This suggests the efficiency of the conditional test with a Markov
basis especially for a sparse table like Table \ref{tab:3}.

In \cite{subtable} we gave a necessary and sufficient condition on the
subtable $S$ so that the set of square-free moves of degree two forms
a Markov basis for $S$.  For a general $S$ it seems to be difficult to
explicitly describe a Markov basis.  For the diagonal $S$ the Markov
basis in Theorem \ref{thm:main} turned out to be relatively simple.
It would be helpful to consider some other special type of $S$ in
order to understand Markov bases for totally general $S$.

We have stated Theorem \ref{thm:main} for the case that $S$ contains
all the diagonal elements $(i,i)$, $i=1, \dots, \min(R,C)$.  Actually
our proof shows that our result can be generalized to $S$ which is a
subset of the diagonal cells.  Furthermore we can relabel the rows and
the columns.  Therefore the essential condition for the result in this
paper is that $S$ contains at most one cell in each row and each
column of the $R\times C$ table.

Theorem \ref{thm:main} was stated for the case
$\min(R,C)\ge 3$ and $\max(R,C)\ge 4$.  For smaller tables, we just
omit moves, which can not fit into small tables.  For completeness we
list these cases and give a Markov basis for each case.  
For avoiding triviality, we assume $\min(R,C) \ge 2$.  

\begin{enumerate}
\item $2\times 2$ : CDEM is the same as the saturated model and no
  degrees of freedom is left for the moves
\item $2\times 3$ : Type V moves form a Markov basis.
\item $2\times C$, $C\ge 4$: Moves of Type I, V and VI form a Markov
  basis.
\item $3\times 3$:  Moves of Type II, III and V form a Markov basis.
\end{enumerate}

It may be interesting and important to extend the subtable sum or/and 
diagonal sum problems to higher dimensional tables.  
However this seems to be difficult at this point and is left
for our future studies. 

\section*{Acknowledgment}  

The authors would like to thank Seth Sullivant for pointing out missing
elements in a Markov basis. 
The authors would also like to thank two anonymous referees for constructive
comments and suggestions.

\bibliographystyle{plainnat}
\bibliography{Hara-Takemura-Yoshida}

\begin{thebibliography}{19}
\providecommand{\natexlab}[1]{#1}
\providecommand{\url}[1]{\texttt{#1}}
\expandafter\ifx\csname urlstyle\endcsname\relax
  \providecommand{\doi}[1]{doi: #1}\else
  \providecommand{\doi}{doi: \begingroup \urlstyle{rm}\Url}\fi

\bibitem[Agresti(2002)]{agre:2002}
Alan Agresti.
\newblock \emph{Categorical Data Analysis}.
\newblock John Wiley and Sons, 2nd edition, 2002.

\bibitem[Aoki and Takemura(2003)]{aoki-takemura-2003anz}
Satoshi Aoki and Akimichi Takemura.
\newblock Minimal basis for a connected {M}arkov chain over {$3\times 3\times
  K$} contingency tables with fixed two-dimensional marginals.
\newblock \emph{Aust. N. Z. J. Stat.}, 45\penalty0 (2):\penalty0 229--249,
  2003.
\newblock ISSN 1369-1473.

\bibitem[Aoki and Takemura(2005)]{aoki-takemura-2005jscs}
Satoshi Aoki and Akimichi Takemura.
\newblock Markov chain {M}onte {C}arlo exact tests for incomplete two-way
  contingency table.
\newblock \emph{Journal of Statistical Computation and Simulation}, 75\penalty0
  (10):\penalty0 787--812, 2005.

\bibitem[Bishop et~al.(1975)Bishop, Fienberg, and Holland]{BFH1975}
Yvonne M.~M. Bishop, Stephen~E. Fienberg, and Paul~W. Holland.
\newblock \emph{Discrete Multivariate Analysis: Theory and Practice}.
\newblock The MIT Press, Cambridge, Massachusetts, 1975.

\bibitem[Diaconis and Sturmfels(1998)]{diaconis-sturmfels}
Persi Diaconis and Bernd Sturmfels.
\newblock Algebraic algorithms for sampling from conditional distributions.
\newblock \emph{Ann. Statist.}, 26\penalty0 (1):\penalty0 363--397, 1998.
\newblock ISSN 0090-5364.

\bibitem[Dobra(2003)]{dobra-2003bernoulli}
Adrian Dobra.
\newblock {M}arkov bases for decomposable graphical models.
\newblock \emph{Bernoulli}, 9\penalty0 (6):\penalty0 1093--1108, 2003.
\newblock ISSN 1350-7265.

\bibitem[Dobra and Sullivant(2004)]{dobra-sullivant}
Adrian Dobra and Seth Sullivant.
\newblock A divide-and-conquer algorithm for generating {M}arkov bases of
  multi-way tables.
\newblock \emph{Comput. Statist.}, 19\penalty0 (3):\penalty0 347--366, 2004.
\newblock ISSN 0943-4062.

\bibitem[Gao and Kuriki(2006)]{gao-kuriki-2006}
Wei Gao and Satoshi Kuriki.
\newblock Testing marginal homogeneity against stochastically ordered marginals
  for {$r\times r$} contingency tables.
\newblock \emph{J. Multivariate Anal.}, 97\penalty0 (6):\penalty0 1330--1341,
  2006.
\newblock ISSN 0047-259X.

\bibitem[Geiger et~al.(2006)Geiger, Meek, and Sturmfels]{geiger-meek-sturmfels}
Dan Geiger, Chris Meek, and Bernd Sturmfels.
\newblock On the toric algebra of graphical models.
\newblock \emph{Ann. Statist.}, 34\penalty0 (3):\penalty0 1463--1492, 2006.

\bibitem[Haberman(1988)]{haberman-1988-jasa}
Shelby~J. Haberman.
\newblock A warning on the use of chi-squared statistics with frequency tables
  with small expected cell counts.
\newblock \emph{J. Amer. Statist. Assoc.}, 83\penalty0 (402):\penalty0
  555--560, 1988.
\newblock ISSN 0162-1459.

\bibitem[Hara et~al.(2007{\natexlab{a}})Hara, Aoki, and Takemura]{sizetwo}
Hisayuki Hara, Satoshi Aoki, and Akimichi Takemura.
\newblock Fibers of sample size two of hierarchical models and {M}arkov bases
  of decomposable models for contingency tables, 2007{\natexlab{a}}.
\newblock Preprint. arXiv:math/0701429v1.

\bibitem[Hara et~al.(2007{\natexlab{b}})Hara, Takemura, and Yoshida]{subtable}
Hisayuki Hara, Akimichi Takemura, and Ruriko Yoshida.
\newblock Markov bases for subtable sum problems, 2007{\natexlab{b}}.
\newblock Preprint. arXiv:0708.2312v1. To appear in {\em Journal of Pure and
  Applied Algebra}.

\bibitem[Ho{\c{s}}ten and Sullivant(2002)]{hosten-sullivant}
Serkan Ho{\c{s}}ten and Seth Sullivant.
\newblock Gr\"obner bases and polyhedral geometry of reducible and cyclic
  models.
\newblock \emph{J. Combin. Theory Ser. A}, 100\penalty0 (2):\penalty0 277--301,
  2002.
\newblock ISSN 0097-3165.

\bibitem[Huber et~al.(2006)Huber, Chen, Dobra, and Nicholas]{Huber-et-al-2006}
Mark Huber, Yuguo Chen, Ian Dinwoodie~Adrian Dobra, and Mike Nicholas.
\newblock Monte carlo algorithms for {H}ardy-{W}einberg proportions.
\newblock \emph{Biometrics}, 62:\penalty0 49--53, 2006.

\bibitem[Lehmann and Romano(2005)]{lehmann-Romano}
E.~L. Lehmann and Joseph~P. Romano.
\newblock \emph{Testing statistical hypotheses}.
\newblock Springer Texts in Statistics. Springer, New York, third edition,
  2005.
\newblock ISBN 0-387-98864-5.

\bibitem[Rapallo(2006)]{Rapallo-2006}
Fabio Rapallo.
\newblock Markov bases and structural zeros.
\newblock \emph{Journal of Symbolic Computation}, 41:\penalty0 164--172, 2006.

\bibitem[Sturmfels(1996)]{sturmfels1996}
Bernd Sturmfels.
\newblock \emph{Gr\"obner Bases and Convex Polytopes}, volume~8 of
  \emph{University Lecture Series}.
\newblock American Mathematical Society, Providence, RI, 1996.
\newblock ISBN 0-8218-0487-1.

\bibitem[Takemura and Aoki(2005)]{takemura-aoki-2005bernoulli}
Akimichi Takemura and Satoshi Aoki.
\newblock Distance reducing {M}arkov bases for sampling from a discrete sample
  space.
\newblock \emph{Bernoulli}, 11\penalty0 (5):\penalty0 793--813, 2005.

\bibitem[Tomizawa(2006)]{tomizawa06}
Sadao Tomizawa.
\newblock Analysis of square contingency tables in statistics.
\newblock \emph{S\=ugaku}, 58\penalty0 (3):\penalty0 263--287, 2006.
\newblock ISSN 0039-470X.

\end{thebibliography}
\end{document}